\documentclass[runningheads,a4paper]{llncs}

\usepackage{ltl}
\usepackage{wrapfig}
\usepackage{booktabs}
\usepackage{tikz}
\usetikzlibrary{calc}
\usepackage{color}
\usepackage{scalefnt}
\usepackage{makecell}
\usepackage{array}
\usepackage{multirow}

\usepackage{amssymb}
\usepackage{graphicx}
\usepackage[noend]{algorithmic}
\usepackage{algorithm2e}
\usepackage{stmaryrd}

\usepackage{url}
\urldef{\mailsa}\path|{alfred.hofmann, ursula.barth, ingrid.haas, frank.holzwarth,|
\urldef{\mailsb}\path|anna.kramer, leonie.kunz, christine.reiss, nicole.sator,|
\urldef{\mailsc}\path|erika.siebert-cole, peter.strasser, lncs}@springer.com|

\newcommand{\todo}[1]{}

\newcommand{\AP}{\textit{AP} }

\newcommand{\A}[0]{\mathcal A}
\newcommand{\T}[0]{\mathcal T}
\newcommand{\N}[0]{\mathcal N}
\newcommand{\D}[0]{\mathcal D}
\newcommand{\E}[0]{\mathcal E}

\newcommand{\nats}{\mathbb{N}}

\newcommand{\outcomes}[1]{L(#1)}
\newcommand{\outcomesE}[1]{L_{\E}(#1)}

\newcommand{\rel}{\mathit{rel}}
\newcommand{\gr}{\mathit{g}}

\begin{document}

\mainmatter  

\title{Synthesizing Approximate Implementations for Unrealizable Specifications\thanks{This work was partially supported by the German Research Foundation (DFG) as part of the Collaborative Research Center ``Foundations of Perspicuous Software Systems'' (TRR 248, 389792660), and by the European Research Council (ERC) Grant OSARES (No. 683300).}
}

\titlerunning{Synthesizing Approximate Implementations}

%
%
\author{Rayna Dimitrova \inst{1}, Bernd Finkbeiner \inst{2} \and Hazem Torfah \inst{2}}
\authorrunning{}

\institute{University of Leicester \and Saarland University}

%
%

\toctitle{Lecture Notes in Computer Science}
\tocauthor{Authors' Instructions}
\maketitle

\begin{abstract}
  The unrealizability of a specification is often due to the assumption that the behavior of the environment is unrestricted. In this paper, we present algorithms for synthesis in bounded environments, where the environment can only generate input sequences that are ultimately periodic words (lassos) with finite representations of bounded size.  We provide automata-theoretic and symbolic approaches for solving this synthesis problem, and also study the synthesis of approximative implementations from unrealizable specifications. Such implementations may violate the specification in general, but are guaranteed to satisfy the specification on at least a specified portion of the bounded-size lassos. We evaluate the algorithms on different arbiter specifications.  \looseness=-1
\end{abstract}

\section{Introduction}\label{sec:intro}
The objective of reactive synthesis is to automatically construct an implementation of a reactive system from a high-level specification of its desired behaviour. While this idea holds a great promise, applying synthesis in practice often faces significant challenges. One of the main hurdles is that the system designer has to provide the right formal specification, which is often a difficult task~\cite{SpecificationBottleneckRozier16}. In particular, since the system being synthesized is required to satisfy its requirements against all possible environments allowed by the specification, accurately capturing the designer's knowledge about the environment in which the system will execute is crucial for being able to successfully synthesize an implementation. 

Traditionally, environment assumptions are included in the specification, usually given as a temporal logic formula. There are, however less explored ways of incorporating information about the environment, one of which is to consider a \emph{bound on the size of the environment}, that is, a bound on the size of the state space of a transition system that describes the possible environment behaviors. Restricting the space of possible environments can render an unrealizable specification into a realizable one. The temporal synthesis under such bounded environments was first studied in~\cite{BoundedSystemsAndEnvironments}, where the authors extensively study the problem, in several versions, from the complexity-theoretic point of view.
\looseness=-1

In this paper, we follow a similar avenue of providing environment assumptions. However, instead of bounding the size of the state space of the environment, we associate a bound with the sequences of values of input signals produced by the environment. The infinite input sequences produced by a finite-state environment which interacts with a finite state system are ultimately periodic, and thus, each such infinite sequence $\sigma \in \Sigma_I^\omega$, over the input alphabet $\Sigma_I$, can be represented as a \emph{lasso}, which is  a pair $(u,v)$ of finite words $u\in \Sigma^{*}_I$ and $v \in \Sigma^{+}_I$, such that $\sigma=u\cdot v^\omega$. It is the length of such sequences that we consider a bound on. More precisely, given a bound $k \in \nats$, we consider the language of all infinite sequences sequences of inputs that can be represented by a lasso $(u,v)$ with $|u \cdot v| = k$. The goal of the \emph{synthesis of lasso precise implementations} is then to synthesize a system for which each execution resulting from a sequence of environment inputs in that language, satisfies a given linear temporal specification.

As an example, consider an arbiter serving two client processes. Each client issues a request when it wants to access a shared resource, and keeps the request signal up until it is done using the resource. The goal of the arbiter is to ensure the classical mutual exclusion property, by not granting access to the two clients simultaneously. The arbiter has to also ensure that each client request is eventually granted. This, however, is difficult since, first, a client might gain access to the resource and never lower the request signal, and second, the arbiter is not allowed to take away a grant unless the request has been set to false, or the client never sets the request to false in the future (the client has become unresponsive). The last two requirements together make the specification unrealizable, as the arbiter has no way of determining if a client has become unresponsive, or will lower the request signal in the future. If, however, the length of the lassos of the input sequences is bounded, then, after a sufficient number of steps, the arbiter can assume that if the request has not been set to false, then it will not be lowered in the future either, as the sequence of inputs must already have run at least once through it's period that will be ultimately repeated from that point  on. \looseness=-1

Formally, we can express the requirements on the arbiter in Linear Temporal Logic (LTL) as follows. There is one input variable $r_i$ (for {\it request}) and one output variable $\gr_i$ (for grant) associated with each client.
The specification is then given as the conjunction
$\varphi_=\varphi_{mutex} \wedge \varphi_{resp} \wedge \varphi_{rel}$
where we use the LTL operators Next $\LTLnext$, Globally $\LTLglobally$ and Eventually  $\LTLfinally$  to define the requirements
\[
\begin{array}{lll}
\varphi_{mutex} & = & \LTLglobally \neg(g_1 \wedge g_2),\\
\varphi_{resp} & = & \LTLglobally \bigwedge_{i=1}^2(r_i \rightarrow \LTLfinally g_i),\\
\varphi_{rel} & = & \LTLglobally \bigwedge_{i=1}^2(\gr_i \wedge  r_i \wedge (\LTLfinally \neg r_i) \rightarrow \LTLnext g_i).\\
\end{array}
\]
Due to the requirement to not revoke grants stated in $\varphi_\rel$, the specification $\varphi$ is unrealizable (that is, there exists no implementation for the arbiter process). For any bound $k$ on the length of the input lassos, however, $\varphi$ is realizable. More precisely, there exists an implementation in which once client $i$ has not lowered the request signal for $k$ consecutive steps, the variable $g_i$  is set to false.

This example shows that when the system designer has knowledge about the resources available to the environment processes, taking this knowledge into account can enable us to synthesize a system that is correct under this assumption.
 
In this paper we formally define the synthesis problem for \emph{lasso-precise implementations}, that is, implementations that are correct for input lassos of bounded size, and describe an automata-theoretic approach to this synthesis problem. We also consider the synthesis of \emph{lasso-precise implementations of bounded size},  and provide a symbolic synthesis algorithm based on quantified Boolean satisfiability. 

Bounding the size of the input lassos can render some unrealizable specifications realizable, but, similarly to bounding the size of the environment, comes at the price of higher computational complexity. To alleviate this problem, we further study the synthesis of \emph{approximate implementations}, where we relax the synthesis problem further, and only require that for a given $\epsilon>0$ the ratio of input lassos of a given size for which the specification is satisfied, to the total number of input lassos of that size is at least $1-\epsilon$. We then propose an \emph{approximate synthesis method} based on maximum  model counting for Boolean formulas~\cite{Fremont:EECS-2016-169}. The benefits of the approximate approach are  two-fold. Firstly, it can often deliver high-quality approximate solutions more efficiently than the lasso-precise synthesis method, and secondly, even when the specification is still unrealizable for a given lasso bound, we might be able to synthesize an implementation that is correct for a given fraction of the possible input lassos.

The rest of the paper is organized as follows. In Section 2 we discuss related work on environment assumptions in synthesis. In Section 3 we provide preliminaries on linear temporal properties and omega-automata. In Section 4 we define the synthesis problem for lasso-precise implementations, and  describe an automata-theoretic synthesis algorithm. In Section 5 we study the synthesis of lasso-precise implementations of bounded size, and provide a reduction to quantified Boolean satisfiability. In Section 6 we define the approximate version of the problem, and give a synthesis procedure based on maximum model counting. Finally, in Section 7 we present experimental results, and conclude  in Section 8.

\section{Related Work}
Providing good-quality environment specifications (typically in the form of assumptions on the allowed behaviours of the environment) is crucial for the synthesis of implementations from high-level specifications. 
Formal specifications, and thus also environment assumptions, are often hard to get right, and have been identified as one of the bottlenecks in formal methods and autonomy~\cite{SpecificationBottleneckRozier16}.
It is therefore not surprising, that there is a plethora of approaches addressing the problem of how to revise inadequate environment assumptions in the cases when these are the cause of unrealizability of the system requirements.

Most approaches in this direction build upon the idea of analyzing the cause of unrealizability of the specification and extracting assumptions that  help eliminate this cause. The method proposed in~\cite{EnvironmentAssumptions} uses the game graph that is used to answer the realizability question in order to construct a B\"uchi automaton representing a minimal assumption that makes the specification realizable. The authors of~\cite{MiningAssumptions} provide an alternative approach where the environment assumptions are gradually strengthened based on counterstrategies for the environment. The key ingredient for this approach is using a library of specification templates and user scenarios for the mining of assumptions, in order to generate good-quality assumptions. A similar approach is used in~\cite{CounterStrategyAssumptionsAlurMT13}, where, however, assumption patterns are synthesized directly from the counterstrategy without the need for the user to provide patterns. A different line of work focuses on  giving feedback to the user or specification designer about the reason for unrealizability, so that they can, if possible, revise the specification accordingly. The key challenge adressed there lies in providing easy-to-understand feedback to users, which relies on finding a minimal cause for why the requirements are not achievable  and generating a natural language explanation of this cause~\cite{ExplainingUnachievableRamanLFLMK13}.

In the above mentioned approaches, assumptions are provided or constructed in the form of a temporal logic formula or an omega-automaton. Thus, it is on the one hand often difficult for specification designers to specify the right assumptions, and on the other hand  special care has to be taken by the assumption generation procedures to ensure that the constructed assumptions are simple enough for the user to understand and evaluate. The work~\cite{BoundedSystemsAndEnvironments} takes a different route, by making assumptions about the \emph{size} of the environment. That is, including as an additional  parameter to the synthesis problem a bound on the state space of the environment. Similarly to temporal logic assumptions, this relaxation of the synthesis problem can render unrealizable specifications into realizable ones. From the system designer point of view, however, it might be significantly easier to estimate the size of environments that are feasible in practice than to express the implications of this additional information in a temporal logic formula. 
In this paper we take a similar route to~\cite{BoundedSystemsAndEnvironments}, and consider a bound on the cyclic structures in the environment's behaviour. Thus, the closest to our work is the temporal synthesis for bounded environments studied in~\cite{BoundedSystemsAndEnvironments}. In fact, we show that the synthesis problem for lasso-precise implementations and the synthesis problem under bounded environments can be reduced to each other. However, while the focus in~\cite{BoundedSystemsAndEnvironments} is on the computational complexity of the bounded synthesis problems, here we provide both automata-theoretic, as well as symbolic approaches for solving the synthesis problem for environments with bounded lassos. We further consider an \emph{approximate version of this synthesis problem}. The benefits of using approximation are two-fold. Firstly, as shown in~\cite{BoundedSystemsAndEnvironments}, while bounding the environment can make some specifications realizable, this comes at a high computational complexity price. In this case, approximation might be able to provide solutions of sufficient quality more efficiently. Furthermore,  even after bounding the environment's input behaviours, the specification might still remain unrealizable, in which case we would like to satisfy the requirements for as many input lassos as possible. In that sense, we get closer to synthesis methods for probabilistic temporal properties in probabilistic environments~\cite{ProbabilisticSystemsKwiatkowskaP13}. However, we consider non-probabilistic environments (i.e., all possible inputs are equally likely), and provide probabilistic guarantees with desired confidence by employing  maximum model counting techniques. Maximum model counting has previously been used for the synthesis of approximate non-reactive programs~\cite{Fremont:EECS-2016-169}. Here, on the other hand we are concerned with the synthesis of reactive systems from temporal specifications.

Bounding the size of the synthesized system implementation is a  complementary restriction of the synthesis problem, which has attracted a lot of attention in recent years~\cite{boundedSynthesis}. The computational complexity of the synthesis problem when both the system's and the environment's size is bounded has been studied in~\cite{BoundedSystemsAndEnvironments}. In this paper we provide a
symbolic synthesis procedure for bounded synthesis of lasso-precise implementations based on quantified Boolean satisfiability.

\section{Preliminaries}\label{sec:preliminaries}
We now  recall definitions and notation from formal languages and automata, and notions from reactive synthesis such as implementation and environment.

\paragraph{Linear-time Properties and Lassos.} A \emph{linear-time property} $\varphi$  over an alphabet~$\Sigma$ is a set of infinite words $\varphi \subseteq \Sigma ^\omega$. 
Elements of $\varphi$ are called \emph{models} of $\varphi$.
A \emph{lasso} of length $k$ over an alphabet $\Sigma$ is a pair $(u,v)$ of finite words $u\in \Sigma^{*}$ and $v \in \Sigma^{+}$ with  $|u\cdot v|~=k$ that induces the ultimately periodic word  $u \cdot v^\omega$. We call $u\cdot v$ the \emph{base} of the lasso or ultimately periodic word, and $k$ the \emph{length} of the lasso.

If a word $w\in \Sigma^*$ is a prefix of a word $\sigma \in \Sigma^*\cup\Sigma^\omega$, we write $w< \sigma$. For a language $L \subseteq \Sigma^* \cup \Sigma^\omega$, we define  $\mathit{Prefix}(L) = \{w \in \Sigma^* \mid \exists \sigma \in L: w < \sigma\}$  is the set of all finite words that are prefixes of words in $L$.\looseness=-1

\paragraph{Implementations.} We represent implementations as \textit{labeled transition systems}. Let $I$  and $O$ be finite sets of \emph{input} and \emph{output atomic propositions} respectively. A $2^O$-labeled $2^I$-transition system is a tuple $\mathcal T = (T,t_0,\tau,o)$, consisting of a finite set of states~$T$, an initial state $t_0\in T$, a transition function $\tau\colon T\times 2^I \rightarrow T$, and a labeling function $o\colon T\rightarrow 2^O$. 
We denote by $|\T|$ the size of an implementation $\T$, defined as $|\T| = |T|$.
A \textit{path} in $\mathcal T$ is a sequence~$\pi\colon \mathbb N \rightarrow T \times 2^I$  of states and inputs that follows the transition function, i.e., for all $i \in \mathbb N$  if $ \pi(i) = (t_i,e_i)$ and $\pi(i + 1) = (t_{i+1}, e_{i+1})$, then $t_{i+1} = \tau(t_i, e_i)$. We call a path \emph{initial} if it starts with the initial state: $\pi(0) = (t_0 , e )$ for some $e \in 2^I$.  
For an initial path $\pi$, we call the sequence $\sigma_{\pi}\colon i \mapsto (o(t_i) \cup e_i) \in (2^{I \cup O})^\omega$ the \emph{trace} of $\pi$.  We call the set of traces of a transition system $\mathcal T$ the \emph{language of $\mathcal T$}, denoted $L(\mathcal T)$. 

\emph{Finite-state environments} can be represented as labelled transition systems in a similar way, with the difference that the inputs are the outputs of the implementation, and the states of the environment are labelled with inputs for the implementation. More precisely, a finite-state environment is a $2^I$-labeled $2^O$-transition system $\mathcal E = (E,s_0,\rho,\iota)$. The composition of an implementation $\T$ and an environment $\E$ results in a set of traces of  $\T$, which we denote $L_{\E}(\mathcal T)$,  where
$\sigma=\sigma_0\sigma_1\ldots  \in L_{\E}(\T)$ if and only if $\sigma \in L(\mathcal T)$ and there exists an initial path $s_0s_1\ldots$ in $\E$ such that for all $i \in \nats$, $s_{i+1} = \rho(s_i,\sigma_{i+1}\cap O)$ and $\sigma_i \cap I = \iota(s_i)$.

\paragraph{Linear-time Temporal Logic.}
We specify properties of reactive systems (implementations) as formulas in Linear-time Temporal Logic (LTL) \cite{LTL}. We consider the usual temporal operators Next $\LTLcircle$, Until $\LTLuntil $, and the derived operators Release $\LTLrelease$, which is the dual operator of $\LTLuntil$, Eventually $\LTLdiamond$ and Globally~$\LTLsquare$.
LTL formulas are defined over a set of atomic propositions $\AP$.
We denote the satisfaction of an LTL formula $\varphi$ by an infinite sequence $\sigma \in (2^{AP})^\omega $ of valuations of the atomic propositions  by $\sigma \models \varphi$ and call $\sigma$ a \emph{model} of $\varphi$. For an LTL formula $\varphi$ we define the language $L(\varphi)$ of $\varphi$ to be the set $\{\sigma \in (2^{AP})^\omega \mid \sigma \models \varphi \}$. 

For a set of atomic propositions $\AP= O \cup I$, we say that a $2^O$-labeled $2^I$-transition system $\mathcal T$ satisfies  an LTL formula $\varphi$, if and only if $L(\T) \subseteq L(\varphi)$, i.e., every trace of $\mathcal T$ satisfies $\varphi$. In this case we call $\mathcal T$ a \emph{model} of $\varphi$, denoted $\T \models \varphi$. If $\T$ satisfies $\varphi$ for an environment $\E$, i.e. $L_\E(\T)\subseteq L(\varphi)$, we write $\T \models_{\E}\varphi$.

For $I \subseteq AP$ and $\sigma \in  (2^{AP})^* \cup (2^{AP})^\omega$, we denote with $\sigma|_{I}$ the projection of $\sigma$ on $I$, obtained by the sequence of valuations of the propositions from $I$ in $\sigma$.

\paragraph{Automata Over Infinite Words.} The automata-theoretic approach to reactive synthesis relies on the fact that an LTL specification can be translated to an automaton over infinite words, or, alternatively, that the specification can be provided directly as such an automaton.
An \emph{alternating parity automaton} over an alphabet $\Sigma$ is a tuple $\mathcal{A} =
(Q,q_0,\delta,\mu)$, where
$Q$ denotes a finite set of states, $Q_0 \subseteq Q$ denotes a set of
initial states, $\delta$ denotes a transition function, and $\mu: Q
\rightarrow C \subset \mathbb{N}$ is a coloring function.  The
transition function $\delta: Q \times \Sigma \rightarrow
\mathbb{B}^+(Q)$ maps a state and an input letter to a
positive Boolean combination of states~\cite{Vardi95Alternating}.

A tree $T$ over a set of directions $D$ is a prefix-closed subset of $D^*$. The empty sequence $\epsilon$ is called the root. The children of a node $n\in T$ are the nodes $\{n\cdot d \in T \mid d\in D\}$.  A $\Sigma$-labeled tree is a pair $(T ,l)$, where $l:T \rightarrow \Sigma$ is the labeling function.
A \emph{run} of $\mathcal{A} = (Q,q_0,\delta,\mu)$ on an infinite word $\sigma = \alpha_0\alpha_1\dots \in \Sigma^\omega$ is a $Q$-labeled tree $(T,l)$ that satisfies the following constraints:
(1) $l(\epsilon) = q_0$, and (2) for all $n \in T$, if $l(n)=q$, then $\{ l(n')\mid  n' \text{ is a child of } n\}$ satisfies $\delta(q,\alpha_{|n|})$.

A run tree is \emph{accepting} if every branch either hits a \emph{true} transition or is an infinite branch $n_0n_1n_2 \dots \in T$, and the sequence
$l(n_0)l(n_1)l(n_2)\dots  $ satisfies the
\emph{parity condition}, which requires that the highest color occurring
infinitely often in the sequence $\mu(l(n_0))\mu(l(n_1))\mu(l(n_2))\dots \in \mathbb
N^\omega$ is even.  An infinite word $\sigma$ is accepted by an automaton $\A$ if there exists an accepting run of $\A$ on $\sigma$.
The set of infinite words accepted by  $\mathcal{A}$ is
called its \emph{language}, denoted $L(\mathcal{A})$.

A \emph{nondeterministic} automaton is a special alternating
automaton, where for all states $q$ and input letters $\alpha$, $\delta(q,\alpha)$ is a disjunction.
An alternating automaton is called \emph{universal} if, for all states $q$ and input letters $\alpha$, $\delta(q,\alpha)$ is a conjunction.
A  universal and nondeterministic automaton is called \emph{deterministic}.

A parity automaton is called a \emph{B\"uchi} automaton if and only if the image of
$\mu$ is contained in $\{1,2\}$, a \emph{co-B\"uchi} automaton if and only if
the image of $\alpha$ is contained in $\{0,1\}$.
B\"uchi and co-B\"uchi
automata are denoted by $(Q,Q_0,\delta,F)$, where
$F\subseteq Q$ denotes the states with the higher color. 
A run graph
of a B\"uchi automaton is thus accepting if, on every infinite path,
there are infinitely many visits to states in $F$; a run graph of a co-B\"uchi
automaton is accepting if, on every path, there are only finitely many visits to  states in $F$.

The next theorem states the relation between LTL and alternating B\"uchi automata, namely that every LTL formula $\varphi$ can be translated to an alternating B\"uchi automaton with the same language and size linear in the length of $\varphi$.
\begin{theorem}{\emph{\cite{alternatingAutomataLTL}}} For every LTL formula $\varphi$ there is an alternating B\"uchi automaton $\A$ of size $O(|\varphi|)$ with $L(\A) = L(\varphi)$, where $|\varphi|$ is the length of $\varphi$.
\end{theorem}

\paragraph{Automata Over Finite Words.} We also use automata over finite words as acceptors for languages consisting of prefixes of traces. A nondeterministic finite automaton over an alphabet $\Sigma$ is a tuple $\A = (Q,Q_0,\delta,F)$, where $Q$ and $Q_0 \subseteq Q$ are again the states and initial states respectively, $\delta : Q \times \Sigma \to 2^Q$ is the transition function and $F$ is the set of accepting states. A run on a word $a_1\ldots a_n$ is a sequence of states $q_0q_1\ldots q_n$, where $q_0 \in Q_0$ and $q_{i+1} \in \delta(q_i,a_i)$. The run is accepting if $q_n \in F$. Deterministic finite automata are defined similarly with the difference that there is a single initial state $q_0$, and that the transition function is of the form $\delta: Q \times \Sigma \to Q$. As usual, we denote the set of words accepted by a nondeterministic or deterministic finite automaton $\A$ by $L(\A)$.
\section{Synthesis of Lasso-precise Implementations}
\label{sec:lassoPrecise}
In this section we first define the synthesis problem for environments producing input sequences representable as lassos of length bounded by a given number. We then provide an automata-theoretic algorithm for this synthesis problem.

\subsection{Lasso-precise implementations}
We begin by formally defining the language of  sequences of input values representable by lassos of a given length $k$. For the rest of the section, we consider linear-time properties defined over a set of atomic propositions $\AP$. The subset $I\subseteq\AP$ consists of the input atomic propositions controlled by the environment.

\begin{definition}[Bounded Model Languages]
	Let $\varphi$ be a  linear-time property over a set of atomic propositions $\AP$, let $\Sigma = 2^{\AP}$, and let $I \subseteq \AP$.
	
	 We say that an infinite word $\sigma \in \Sigma^\omega$ is an $I$-$k$-model of $\varphi$, for a bound $k\in \nats$, if and only if there are words $u \in (2^I)^*$ and $v \in (2^I)^+$ such that $|u\cdot v| = k $ and $\sigma|_{I} = u\cdot v^\omega$. 
	The language of $I$-$k$-models of the property $\varphi$ is defined by the set  $L_k^{I}(\varphi)=\{\sigma \in \Sigma^\omega \mid \sigma \text{ is a } I\text{-}k\text{-model of } \varphi \}$. 
\end{definition}

Note that a model of $\varphi$ might be induced by lassos of different length and by more than one lasso of the same length, e.g, $a^\omega$ is induced by $(a, a)$ and $(\epsilon, a a)$. The next lemma establishes that if a model of $\varphi$ can be represented by a lasso of length $k$ then it can also be represented by a lasso of any larger length. 

\begin{lemma}\label{lemma:nfa-prefix}
	For a linear-time property $\varphi$ over $\Sigma=2^{\AP}$, subset $I \subseteq \AP$ of atomic propositions, and bound $k\in \nats$, we have $L_k^I(\varphi)\subseteq L_{k'}^I(\varphi)$ for all $k'> k. $\looseness=-1
\end{lemma}
\begin{proof}
	Let $\sigma \in L_k^{I}(\varphi)$. Then,  $\sigma \models \varphi$ and there exists $(u,v) \in (2^{I})^*\times (2^{I})^+$ such that $|u\cdot v|=k$ and  $\sigma|_I= u\cdot v^\omega$. Let $v=v_1\dots v_k$. Since $ u\cdot v_1 (v_2\dots v_k v_1)^\omega = u\cdot (v_1\dots v_k)^\omega = \sigma|_{I}$, we have  $\sigma \in L_{k+1}^{I}(\varphi)$. The claim  follows by induction. 
	\qed 
\end{proof}

Using the definition of $I$-$k$-models, the language of infinite sequences of environment inputs representable by lassos of length $k$ can be expressed as $L_k^I(\Sigma^\omega)$.

\begin{definition}[$k$-lasso-precise Implementations]
For a linear-time property $\varphi$ over $\Sigma=2^{\AP}$, subset $I \subseteq \AP$ of atomic propositions, and bound $k\in \nats$, we say that a transition system $\T$ is a $k$-lasso-precise implementation of $\varphi$, denoted $\T \models_{k,I} \varphi$, if it holds that $L_k^I(L(\T)) \subseteq \varphi$.
\end{definition}

That is, in a $k$-lasso-precise implementation $\T$ all the traces of $\T$ that belong to the language $L_k^I(\Sigma^\omega)$ are $I$-$k$-models of the specification $\varphi$. 

\bigskip

{\bf Problem definition: Synthesis of lasso-precise implementations.} 

\noindent
Given a linear-time property $\varphi$ over atomic propositions $\AP$ with input atomic propositions $I$, and given a bound $k \in \nats$, construct an implementation $\T$ such that  $\T \models_{k,I} \varphi$, or determine that such an implementation does not exist.

Another way to bound the behaviour of the environment is to consider a bound on the size of its state space. The \emph{synthesis problem for bounded environments} asks  for a given linear temporal property $\varphi$ and a bound $k \in \nats$ to synthesize a transition system $\T$ such that for every possible environment $\E$ of size at most $k$, the transition system $\T$ satisfies $\varphi$ under environment $\E$, i.e., $T\models_{\E} \varphi$.\looseness=-1

We now establish the relationship between the synthesis of lasso-precise implementations and synthesis under bounded environments. Intuitively,  the two synthesis problems can be reduced to each other since an environment of a given size, interacting with a given implementation,  can only produce ultimately periodic sequences of inputs representable by lassos of length determined by the sizes of the environment and the implementation. This intuition is formalized in the following proposition, stating the connection between the two  problems.

\begin{proposition}
	Given a specification $\varphi$  over a set of atomic propositions $\AP$ with subset $I \subseteq AP$ of atomic propositions controlled by the environment,  and a bound $k \in \nats$,  for every transition system $\T$ the following statements hold:
	\begin{itemize}
	\item[(1)] If $\T \models_{\E} \varphi$ for all environments $\E$ of size at most $k$, then $\T \models_{k,I} \varphi$. 
	\item[(2)] If $\T \models_{k\cdot|\T|,I} \varphi$, then $\T \models_{\E} \varphi$ for all environments $\E$ of size at most $k$.
	\end{itemize}
\end{proposition}
\begin{proof}
For {\it (1)}, let $\T$ be a transition system  such that $\T \models_{\E} \varphi$ for all environments $\E$ of size at most $k$. Assume, for the sake of contradiction, that $\T \not\models_{k,I} \varphi$. Thus, that there exists a word $\sigma \in \outcomes{\mathcal T}$, such that $\sigma\in L_{k}^{I}(\Sigma^\omega)$ and $\sigma \not\models\varphi$.

Since $\sigma\in L_{k}^{I}(\Sigma^\omega)$, we can construct an environment $\E$ of size at most $k$ that produces the sequence of inputs $\sigma|_{I}$.  Since $\E$ is of size at most $k$, we have that $\T \models_{\E} \varphi$. Thus, since $\sigma\in \outcomesE{\T}$, we have $\sigma \models \varphi$, which is a contradiction.

For {\it (2)}, let $\T$ be a transition system such that $\T \models_{k\cdot|\T|,I} \varphi$. Assume, for the sake of contradiction that there exists an environment $\E$ of size at most $k$ such that $\T \not\models_{\E} \varphi$.  Since 
$\T \not\models_{\E} \varphi$, there exists $\sigma\in \outcomesE{\T}$ such that $\sigma \not\models \varphi$. As the number of states of $\E$ is at most $k$, the input sequences it generates can be represented as lassos of size  $k\cdot|\T|$. Thus, $\sigma\in L_{k\cdot|\T|}^{I}(\Sigma^\omega)$. This is a contradiction with the choice of $\T$, according to which  $\T \models_{k\cdot|\T|,I} \varphi$.
\qed
\end{proof}

\subsection{Automata-theoretic synthesis of lasso-precise implementations}

We now provide an automata-theoretic algorithm for the synthesis of lasso-precise implementations. The underlying idea of this approach is to first construct an automaton over finite traces that accepts all finite prefixes of traces in $L_k^I(\Sigma^\omega)$. Then, combining this automaton and an automaton representing the property $\varphi$ we can construct an automaton whose language is non-empty if and only if there exists an $k$-lasso-precise implementation of $\varphi$.

The next theorem presents the construction of a deterministic finite automaton for the language $\mathit{Prefix}(L_k^{I}(\Sigma^\omega))$.
\begin{theorem}\label{thm:dfa-prefix}
	For any set $\AP$ of atomic propositions, subset $I\subseteq \AP$, and bound $k\in \nats$ there is a deterministic finite automaton $\A_k$ over alphabet $\Sigma=2^{\AP}$, with size $(2^{|I|}+1)^k \cdot (k+1)^k$, such that $L(\A_k)= \{w \in \Sigma^* \mid \exists \sigma \in L_k^{I}(\Sigma^\omega). ~w <\sigma\}.$
\end{theorem}
\paragraph{Idea \& Construction.}
For given $k\in \nats$ we first define an automaton $\widehat\A_k = (Q,q_0,\delta,F)$ over $\widehat\Sigma = 2^I$, such that  $L(\widehat \A_k)= \{\widehat w \in \widehat\Sigma^* \mid \exists \widehat\sigma \in L_k^{I}(\widehat\Sigma^\omega). ~\widehat w <\widehat \sigma\}.$ That, is $L(\widehat \A_k)$ is the set of all finite prefixes of infinite words over $\widehat\Sigma$ that can be represented by a lasso of length $k$. We can then define the automaton $\A_k$ as the automaton that for each $w \in \Sigma^*$ simulates $\widehat{\A}_k$ on the projection $w|_{I}$ of $w$.
We define the automaton $\widehat\A_k = (Q,q_0,\delta,F)$ such that
\begin{itemize}
	\item $ Q = (\widehat\Sigma\cup \{\#\})^k \times \{-,1,\dots,k\}^k$,
	\item $q_0 = (\#^k,(1,2,\dots,k))$,
	\item $\delta(q,\alpha)= \begin{cases}
          						  (w\cdot\alpha\cdot \#^{m-1},t) & \text{ if }q=(w\cdot\#^m,t) 
          						 \text{ where } 1\leq m\leq k, \\
          						  &\;\;\;\; w \in \widehat\Sigma^{(k-m)},~ t \in \{\--,1,\dots,k\}^k\\&\\
          						  (w,(i'_1,\dots,i'_k)) &  \text{ if } q = (w,(i_1,\dots,i_k))  \text{ where } w\in \widehat\Sigma^k, \text{ and} \\
          						  & \;\;i'_j = \begin{cases}
          						  			- & i_j \leq k \wedge w(i_j) \not = \alpha  \text{ or } i_j = -\\
          						  			\\
          						  			i_j+1 & i_j<k \wedge w(i_j)=\alpha\\
          						  			\\
          						  			j & i_j=k \wedge w(i_j) =\alpha
          						  		   \end{cases} 
   						     \end{cases}
   		  $
   	\item $F = Q\setminus \{(w,(-,\dots,-)) \mid w \in \widehat\Sigma^k\}$.
\end{itemize}
\begin{proof}
States of the form $(w\cdot\alpha\cdot \#^{m},t)$ with $m \geq 1$ store the portion of the input word read so far, for input words of length smaller than $k$. In states of this form we have $t=(1,2,\ldots, k)$, which implies that all such states are accepting. In turn, this means that $\A_k$ accepts all words of length smaller or equal to $k$.  This is justified by the fact that, each word of length smaller or equal to $k$ is a prefix of an infinite word  in $L_k^{I}(\widehat\Sigma^\omega)$, obtained by repeating the prefix infinitely often. 
Now, let us consider words of length greater than $k$.

In states of the form $(u,(i_1,\ldots,i_k))$, with $u \in \widehat\Sigma^*$, the word $u$ stores the first $k$ letters of the input word. Intuitively, the tuple $(i_1,\ldots,i_k)$ stores the information about the loops that are still possible, given the portion of the input word that is read thus far. To see this, let us consider a word $w \in \widehat\Sigma^*$ such that $|w|  = l > k$, and let $q_0 q_1 \dots q_l$  be the run of $\A_k$ on $w$. The state $q_l$ is of the form $q_l = (w(1)\ldots w(k), (i_1^l,\ldots,i_k^l))$. It can be shown by induction on $l$ that for each $j$ we have $i_j^l \neq -$ if and only if $w$ is of the form $w=w'\cdot w''\cdot w'''$ where 
$w' = w(1)\ldots w(j-1)$, $w'' = (w(j)\ldots w(k))^k$ for some $k \geq 0$, and $w''' = (w(j)\ldots w(i_j^l-1))$.  Thus, if $i_j^l \neq -$, then it is possible to have a loop starting at position $j$, and  $i_j^l$ is such that $(w(j)\ldots w(i_j^l-1))$ is the prefix of $w(j)\ldots w(k)$  appearing after the (possibly empty) sequence of repetitions of $w(j)\ldots w(k)$. This means, that if $i_j^l \neq -$, then $w$ is a prefix of the infinite word $w' \cdot (w'')^\omega\in L_k^I(\widehat\Sigma^\omega)$. Therefore, if the run of $\A_k$ on a word $w$ with $|w| > k$ is accepting,  then there exists $\sigma \in L_k^I(\widehat\Sigma^\omega)$ such that $w <\sigma$. 

For the other direction, suppose that for each $j$, we have $i_j^l=-$. Take any $j$, and consider the first position $m$ in the run $q_0 q_1\ldots q_l$ where $i_j^m=-$. By the definition of $\delta$ we have that $w(m) \neq w(i_j^{m-1})$. This means that the prefix $w(1)\ldots w(m)$ cannot be extended to the word $w(1)\ldots w(j-1)(w(j)\ldots w(k))^\omega$. Since for every $j\in \{1,\ldots,k\}$ we can find such a position $m$, it holds that  there does not exist $\sigma \in L_k^I(\widehat\Sigma^\omega)$ such that $ w <\sigma$. This concludes the proof. \qed
\end{proof}

The automaton constructed in the previous theorem has size which is exponential in the length of the lassos. In the next theorem we show that this exponential blow-up is unavoidable. That is, we show that every nondeterministic finite automaton for the language $\mathit{Prefix}(L_k^{I}(\Sigma^\omega))$ is of size at least $2^{\Omega(k)}$.
\begin{theorem}
	For any  bound $k\in\nats$ and sets of atomic propositions $\AP$ and $\emptyset \neq I \subseteq\AP$,  every nondeterministic finite automaton $\N$ over the alphabet $\Sigma=2^{\AP}$ that recognizes  $L = \{w \in \Sigma^* \mid \exists \sigma \in L_k^{I}(\Sigma^\omega). ~w <\sigma\}$ is of size at least $2^{\Omega(k)}$.  
\end{theorem}
\begin{proof}
Let $\N = (Q,Q_0,\delta,F)$ be a nondeterministic finite automaton for $L$. For each $w \in \Sigma^k$, we have that $w\cdot w \in L$. Therefore, for each $w \in \Sigma^k$ there exists at least one accepting run $\rho=q_0 q_1 \ldots q_f$ of $\N$ on $w\cdot w$.  We denote with $q(\rho,m)$ the state $q_m$ that appears at the position indexed $m$ of a run $\rho$. 

Let $a \in 2^I$ be a letter in $2^I$, and let $\Sigma' = \Sigma \setminus \{ a' \in \Sigma \mid a'|_{I} = a\}$. Let $L' \subseteq L$ be the language $L' = \{w \in \Sigma^k \mid  \exists w' \in (\Sigma')^{k-1},a'\in \Sigma:~w = w'\cdot a' \text{ and }a'|_{I}=a\}$. That is, $L'$ consists of the words of length $k$ in which letters $a'$ with $a'|_I=a$ appear in the last position and only in the last position. 

Let us define the set of states
\[Q_k  =\{q(\rho,k) \mid  \exists w \in L':\; \rho \text{ is an accepting run of }\N \text{ on } w\cdot w\}.\]

That is, $Q_k$ consists of the states that appear at position $k$ on some accepting run on some word $w \cdot w $, where $w$ is from $L'$. We will show that  $|Q_k| \geq 2^{k-1}$. 

Assume that this does not hold, i.e., $|Q_k| < 2^{k-1}$. Since $|L'| \geq  2^{k-1}$, this implies that there exist  $w_1,w_2 \in L'$, such that $w_1|_I \neq w_2|_I$ and there exists accepting runs $\rho_1$ and $\rho_2$ of $\N$ on $w_1\cdot w_1$ and $w_2\cdot w_2$ respectively, such that $q(\rho_1,k)  = q(\rho_2,k)$. That is, there must be two words in $L'$ with $w_1|_I \neq w_2|_I$, which have accepting runs on $w_1\cdot w_1$ and $w_2 \cdot w_2$ visiting the same state at position $k$.\looseness=-1

We now construct a run $\rho_{1,2}$ on the word $w_1 \cdot w_2$ that follows $\rho_1$ for the first $k$ steps on $w_1$, ending in state $q(\rho_1,k)$, and from there on follows $\rho_2$ on $w_2$. It is easy to see that $\rho_{1,2}$ is a run on the word $w_1 \cdot w_2$. The run is accepting, since $\rho_2$ is accepting.  This means that $w_1 \cdot w_2 \in L$, which we will show leads to contradiction.

To see this, recall that $w_1 = w_1'\cdot a'$ and $w_2 = w_2' \cdot a''$, and $w_1|_I \neq w_2|_I$, and $a'|_I=a''|_I=a$. Since $w_1\cdot w_2 \in L$, we have that $w_1'\cdot a' \cdot w_2' \cdot a'' <  \sigma$ for some $\sigma \in L_k^{I}(\Sigma^\omega)$. That is, there exists a lasso for some word $\sigma$, and $w_1'\cdot a' \cdot w_2' \cdot a''$ is a prefix of this word. Since $a$ does not appear in $w_2'|_{I}$, this means that the loop in this lasso is the whole word $w_1|_I$, which is not possible, since $w_1|_I \neq w_2|_I$.

 This is a contradiction,  which shows that  $|Q| \geq  |Q_k| \geq 2^{k-1}$. Since $\mathcal N$ was an arbitrary nondeterministic finite automaton for $L$, this implies that the minimal automaton for $L$ has at least $2^{\Omega(k)}$ states, which concludes the proof.\qed	
\end{proof}

Using the automaton from Theorem~\ref{thm:dfa-prefix}, we can transform every property automaton $\A$ into an automaton that accepts words representable by lassos of length less than or equal to $k$ if and only if they are in $L(\A)$, and accepts all words that are not representable by lassos of length less than or equal to $k$.
\begin{theorem}\label{thm:automaton-synth}
	Let $\AP$ be a set of atomic propositions, and let $I \subseteq \AP$. For every (deterministic, nondeterministic or alternating) parity automaton $\A$ over  $\Sigma=2^{\AP}$, and $k \in \nats$, there is a (deterministic, nondeterministic or alternating) parity automaton $\A'$ of size $2^{O(k)} \cdot |\A|$, s.t., $L(\A') = (L_k^I(\Sigma^\omega)\cap L(\A)) \cup (\Sigma^\omega\setminus L_k^I(\Sigma^\omega))$.
\end{theorem}
\begin{proof}
The theorem is a consequence of 	Theorem~\ref{thm:dfa-prefix} established as follows. Let $\A = (Q,Q_0,\delta,\mu)$ be a parity automaton, and let $\D = (\widehat Q,\widehat q_0,\widehat \delta,F)$ be the deterministic finite automaton for bound $k$ defined as in Theorem~\ref{thm:dfa-prefix}. We define the parity automaton 
 $\A = (Q',Q_0',\delta',\mu')$ with the following components:
 \begin{itemize}
	\item $ Q' = (Q \times \widehat Q)$;
	\item $Q_0' = \{(q_0,\widehat q_0) \mid q_0 \in Q_0\}$ (when $\A$ is deterministic $Q_0'$ is a singleton set);
	\item $\delta'((q,\widehat q),\alpha)= \delta(q,\alpha)_{[q'/(q',\widehat{\delta}(\widehat{q},\alpha))]},
   		  $ where $\delta(q,\alpha)_{[q'/(q',\widehat q')]}$ is the Boolean expression obtained from $\delta(q,\alpha)$ by replacing every state $q'$ by the state $(q',\widehat q')$;
   	\item $\mu'((q,\widehat q)) =
   	 \begin{cases}
   	 \mu(q) & \text{ if } \widehat{q} \in F,\\
   	 0 & \text{ if } \widehat{q} \not\in F.\\
   	 \end{cases}$
\end{itemize}
Intuitively, the automaton $\A'$ is constructed as the product of $\A$ and $\D$, where runs entering a state in $\D$ that is not accepting in $\D$ are accepting in $\A'$. To see this, recall from the construction in Theorem~\ref{thm:dfa-prefix}  that once $\D$ enters a state in $\widehat{Q}\setminus\widehat{F}$ it remains in such a state forever. Thus, by setting the color of all states $(q,\widehat q)$ where $\widehat{q} \not\in F$ to $0$, we ensure that words containing a prefix rejected by $\D$ have only runs in which the highest color appearing infinitely often is $0$. Thus, we ensure that all  words that are not representable by lassos of length less than or equal to $k$ are accepted by $\A'$, while words representable by lassos of length less than or equal to $k$ are accepted if and only if they are in $L(\A)$.\qed
\end{proof}

The following theorem is a consequence of the one above, and provides us with an automata-theoretic approach to solving the lasso-precise synthesis problem.

\begin{theorem}[Synthesis] Let $\AP$ be a set of atomic propositions, and  $I \subseteq \AP$ be a subset of $\AP$ consisting of the atomic propositions controlled by the environment. For a specification, given as a deterministic parity automaton $\mathcal P$ over the alphabet $\Sigma  = 2^{\AP}$, and a bound $k \in \nats$, finding an implementation $\T$, such that, $\T \models_{k,I} \mathcal P$ can be done in time polynomial in the size of the automaton $\mathcal P$ and exponential in the bound $k$.  
\end{theorem}

\section{Bounded Synthesis of Lasso-precise Implementations}
\label{sec:boundedLassoPrecise}
For a specification $\varphi$ given as an LTL formula, a bound $n$ on the size of the synthesized implementation and a bound $k$ on the lassos of input sequences, \emph{bounded synthesis of lasso-precise implementations}
searches for an implementation $\T$ of size $n$, such that $\T \models_{k,I} \varphi$. Using the automata constructions in the previous section we can construct a universal co-B\"uchi automaton for the language $L_k^I(\varphi) \cup (\Sigma^\omega\setminus L_k^I(\Sigma^\omega))$ and construct the constraint system as presented in~\cite{boundedSynthesis}. This constraint system is exponential in both $|\varphi|$ and $k$. In the following we show how the problem can be encoded as a quantified Boolean formula of size polynomial in $|\varphi|$ and $k$. 

\begin{theorem}
	For a specification given as an LTL formula $\varphi$, and bounds $k \in \nats$ and $n\in\nats$, there exists a quantified Boolean formula $\phi$, such that, $\phi$ is satisfiable if and only if there is a transition system $\T =(T,t_0,\tau,o)$ of size $n$ with $\T \models_{k,I} \varphi$. The size of $\phi$ is in  $O(|\varphi| + n^2 + k^2)$. The number of variables of $\phi$ is equal to $n\cdot(n\cdot 2^{|I|} +|O|) + k\cdot(|I|+1) +n\cdot k(|O|+n+1)$.
\end{theorem}

\paragraph{Construction.} We encode the bounded synthesis problem in the following quantified Boolean formula: 
\begin{flalign}
	& \exists \{\tau_{t,i,t'} \mid t,t' \in T, i \in 2^I\}. ~\exists \{o_t \mid t\in T, o \in O\}.&\\
	& \forall \{i_j \mid i \in I, 0 \leq j < k\}. ~\forall \{l_j \mid 0 \leq j < k\}. &\\
	& \forall \{o_j \mid o\in O, 0\leq j < n\cdot k\}. &\\
	& \forall \{t_j \mid t \in T, 0 \leq j < n\cdot k\}. &\\
	& \forall \{l'_j \mid 0 \leq j <n\cdot k\}. &\\
	& \varphi_{\text{det}} \wedge (\varphi_{\text{lasso}} \wedge \varphi_{\in \T}^{n,k} \rightarrow \llbracket \varphi \rrbracket_0^{k, n\cdot k})& 
\end{flalign}
which we read as: there is a transition system (1), such that, for all input sequences representable by lassos of length $k$ (2) the corresponding sequence of outputs of the system (3) satisfies $\varphi$. 
The variables introduced in lines (4) and (5) are necessary to encode the  corresponding output for the chosen input lasso. 

An assignment to the variables satisfies the formula in line (6), if it represents a deterministic transition system ($\varphi_{\text{det}}$) in which lassos of length $n\cdot k$ ($\varphi_{\text{lasso} \wedge \varphi_{\in \T}^{n,k}}$)  satisfy the property $\varphi$ ($\llbracket \varphi \rrbracket_0^{(k, n\cdot k)})$). These constraints are defined as follows. 

$\varphi_{\text{det}}$: A transition system is deterministic if for each state $t$ and input $i$ there is exactly one transition $\tau_{t,i,t'}$ to some state $t'$:
	$\bigwedge \limits_{t \in T} ~\bigwedge \limits_{i \in 2^I} ~\bigvee \limits_{t' \in T} (\tau_{t,i,t'} \wedge \bigwedge \limits_{t' \not =t''} \overline{\tau_{t,i,t''}})$.

$\varphi_{\in \T}^{n,k}$: for a certain input lasso of size $k$ we can match a lasso in the system of size at most $n\cdot k$. A lasso of this size in the transition system matches the input lasso if the following constraints are satisfied.   
\begin{align}
		& \bigwedge \limits_{0 \leq j < n\cdot k} ~ \bigwedge \limits_{t \in T}  (t_j \rightarrow \bigwedge \limits_{o\in O}(o_j \leftrightarrow o_{t_j}))  \\
		\wedge & ~t_{00}
		\\
		\wedge & \bigwedge \limits_{0 \leq j <n\cdot k-1} ~\bigwedge \limits_{i \in 2^I} ~\bigwedge \limits_{t,t' \in T} ((\bigwedge \limits_{0\leq j'<k}l_{j'} \rightarrow i_{\Delta(j,k,j')}) \wedge t_j \rightarrow (\tau_{t,i,t'} \leftrightarrow t'_{j+1}) )\\
		\wedge &  \bigwedge \limits_{i \in 2^I,t,t' \in T} ((\bigwedge \limits_{0\leq j'<k} l_{j'}\rightarrow i_{\Delta(n\cdot k-1,k,j')}) \wedge t_{n\cdot k-1} \rightarrow (\tau_{t,i,t'} \leftrightarrow (\bigvee \limits_{0\leq j < n\cdot k} l'_j \wedge t'_{j}) ))
	\end{align}
	Lines (9) and (10) make sure that the chosen lasso follows the guessed transition relation $\tau$. Line (10) handles the loop transition of the lasso, and makes sure that the loop of the lasso follows $\tau$. Line (7) is a necessary requirement in order to match the output produced on the lasso with $\varphi$. If the output variables $o_j$ satisfy the constraint $\llbracket \varphi \rrbracket_0^{(k,n\cdot k)}$, then the lasso satisfies $\varphi$. As the input lasso is smaller than its matching lasso in the system we need to make sure that the indices of the input variables are correct with respect to the chosen loop. This is computed using the function $\Delta$ which is given by: \\
	 	
	 	$\Delta(j,k,j')= \begin{cases}
						    j & \text{ if }j<k,\\
						    ((j-k) \mod (k -j') )+j'& \text{otherwise.}
						   \end{cases}$\\
	
$\varphi_\text{lasso}$: The formula  encodes the additional constraint that exactly one of the loop variables can be true for a given variable valuation.  

$\llbracket \varphi \rrbracket^{k,m}_0$: This constraint encodes the satisfaction  of $\varphi$ on lassos of size $m$. The encoding is similar to the encoding of bounded model checking \cite{bmc}, with the distinction of encoding the satisfaction relation of the atomic propositions, given below. As the inputs run with different indices than the outputs, we again, as in the lines (9) and (10), need to compute the correct indices using the function $\Delta$.
\looseness=-1
 
\begin{center}
			\small
			\begin{tabular}{|c||c|c|}
			\hline
			& $h<m$ & $h=m$ \\
			\hline
			$\llbracket i \rrbracket ^{k,m}_h$ & $ \bigwedge \limits_{0\leq j' <k}(l_{j'} \rightarrow i_{\Delta(h,k,j')})$ & $\bigvee_{j=0}^{m-1} (l'_j \wedge \bigwedge \limits_{0\leq j' <k}(l_{j'} \rightarrow i_{\Delta(j,k,j')}))$ \\
			\hline 
			$\llbracket \neg i \rrbracket ^{k,m}_h$ & $\bigwedge \limits_{0\leq j' <k}(l_{j'} \rightarrow  \neg i_{\Delta(h,k,j')})$& $\bigvee_{j=0}^{m-1} (l'_j \wedge \bigwedge \limits_{0\leq j' <k}(l_{j'} \rightarrow \neg i_{\Delta(j,k,j')}))$ \\
			\hline
			$\llbracket o \rrbracket ^{k,m}_h$ & $o_h$ & $\bigvee_{j=0}^{m-1} (l'_j \wedge o_j)$ \\
			\hline 
			$\llbracket \neg o \rrbracket ^{k,m}_h$ & $\neg o_h$& $\bigvee_{j=0}^{m-1} (l'_j \wedge \neg o_j)$ \\
			\hline
		\end{tabular}
\end{center}

\section{Synthesis of Approximate Implementations}
\label{sec:approx}

In some cases, specifications remain unrealizable even when considered under bounded environments. Nevertheless, one might still be able to construct implementations that satisfy the specification in almost all input sequences of the environment. Consider for example the following  simplified arbiter specification: 
$$\LTLsquare(\overline{w} \rightarrow \LTLcircle \overline{g}) \wedge \LTLsquare (r 	\rightarrow \LTLdiamond g) $$
The specification defines an arbiter that should give grants $g$ upon requests $r$, but is not allowed to provide these grants unless a signal $w$ is true. The specification is unrealizable, because a sequence of inputs where the signal $w$ is always false prevents the arbiter from answering any request. Bounding the environment does not help in this case as a lasso of size 1 already suffices to violate the specification (the one where $w$ is always false). Nevertheless, one can still find reasonable implementations that satisfy the specification for a large fraction of input sequences. In particular, the fraction of input sequences where $w$ remains false forever is less probable.

\begin{definition}[$\epsilon$-$k$-Approximation]
	For a specification $\varphi$, a bound $k$, and an error rate $\epsilon$, we say that a transition system $\T$ approximately satisfies $\varphi$ with an error rate $\epsilon$ for lassos of length at most $k$, denoted by $\T \models_{k,I}^\epsilon \varphi$, if and only if, $\frac{|\{\sigma \mid \sigma \in L^I_k(L(\T)), \sigma \models \varphi\}|}{|L_k^I((2^I)^\omega)|} \ge 1-\epsilon$. We call $\T$ an $\epsilon$-$k$-approximation of $\varphi$.
\end{definition}

\begin{theorem}
	For a specification given as a deterministic parity automaton $P$, a bound $k$ and a error rate  $0 \leq \epsilon \leq 1$, checking whether there is an implementation $\T$, such that, $\T \models_{k,I}^\epsilon P$ can be done in time polynomial in $|P|$ and exponential in $k$.
\end{theorem}
\begin{proof}
	For a given $\epsilon$ and $k$, we construct a nondeterministic parity tree automaton $\cal N$ that accepts all $\epsilon$-$k$-approximations with respect to  $L(P)$. For $\epsilon$, we can compute the minimal number $m$ of lassos from $L_k^I((2^I)^\omega)$ for which an $\epsilon$-$k$-approximation has to satisfy the specification. In its initial state, the automaton $\cal N$ guesses $m$ many lassos and accepts a transition system if it does not violate the specification on any of these lassos. The latter check is done by following the structure of the automaton constructed for $P$ using Theorem~\ref{thm:automaton-synth}. In order to check whether there is an $\epsilon$-$k$-approximation for $P$, we solve the emptiness game of $\cal N$. The size of $\cal N$ is $(2^k)^{m+1} \cdot|P|$.
	\qed
\end{proof}

\subsection{Symbolic Approach}

In the following, we present a symbolic approach for finding $\epsilon$-$k$-approximations based on maximum model counting. We show that we can build a constraint system and apply a maximum model counting algorithm to compute a transition system that satisfies a specification for a maximum number of input sequences. 

\begin{definition}[Maximum Model Counting \cite{Fremont:EECS-2016-169}]
Let $X,Y$ and $Z$ be sets of propositional variables and $\phi$ be a formula over $X,Y$ and $Z$. Let $x$ denote an assignment to $X$, $y$ an assignment to $Y$, and $z$ an assignment to $Z$. The maximum model counting problem for $\phi$ over $X$ and $Y$ is computing a solution for $\max \limits_{x}\#y.\exists z. \phi(x,y,z). $
\end{definition}

For a specification $\varphi$, bounds $k$ and $n$ on the length of the lassos and size of the system, respectively, we can compute an $\epsilon$-$k$-approximation for $\varphi$ by applying a maximum model counting algorithm to the constraint system given below. It encodes transition systems of size $n$ that have an input lasso of length $k$ that satisfies $\varphi$.
\begin{flalign}
	& \exists \{\tau_{t,i,t'} \mid t,t' \in T, i \in 2^I\}.~ \exists \{o_t \mid t \in T, o \in O\}.  & \\
	& \exists \{i_j \mid i \in I, 0 \leq j < k\}. ~\exists\{l_j \mid 0 \leq j < k\}.&\\
	& \exists \{x_j^i \mid x \in I,  0 \leq i,j < k\} &\\
	& \exists \{o_j \mid o \in O, 0 \leq j < n\cdot k\}.\\  
	& \exists \{t_j \mid t \in T, 0 \leq j < n\cdot k\}. &\\
	& \exists \{l'_j \mid 0 \leq j <n\cdot k\}. &\\
	& \varphi_{\text{det}} \wedge \varphi_{\text{lasso}}\wedge \varphi_{\in \T}^{n,k} \wedge  \llbracket \varphi \rrbracket_0^{ k, n\cdot k} \wedge \llbracket k \rrbracket_0 &
\end{flalign}
To check the existence of a $\epsilon$-$k$-approximation, we maximize over the set of assignment to variables that define the transition system (line 11) and count over variables that define input sequences of the environment given by lassos of length $k$. As two input lassos of the same length may induce the same infinite input sequence, we count over auxiliary variables that represent unrollings of the lassos instead of counting over the input propositions themselves (line 13). 

The formulas $\varphi_{\text{det}}$, $\varphi_{\text{lasso}}$, $\varphi_{\in \T}^{n,k}$ and $\llbracket \varphi \rrbracket_0^{k,n\cdot k} $ are defined as in the previous section. 
The formula $\llbracket k \rrbracket _0 $ is defined over that variables in line (13) and makes sure that input lasso that represent the same infinite sequence are not counted twice by unrolling the lasso to size $2k$.

\begin{theorem}
	For a specification given as an LTL formula $\varphi$, and bounds $k$ and $n$, and an error rate $\epsilon$,  the propositional formula $\phi$ defined above is of size $O(|\varphi| + n^2 + k^2)$. The number of variables of $\phi$ is equal to $n\cdot(n\cdot 2^{|I|} +|O|) + k\cdot(k\cdot |I|+ |I|+1) +n\cdot k(|O|+n+1)$.
\end{theorem}

\section{Experimental Results}
\label{sec:experiments}
We implemented the symbolic encodings for the exact and approximate synthesis methods, and evaluated our approach on a bounded version of the greedy arbiter specification given in Section~\ref{sec:intro},  and another specification of a round-robin arbiter. 
The round-robin arbiter is defined by the specification: $$\LTLsquare \LTLdiamond w \rightarrow \LTLsquare \LTLdiamond g_1 \wedge \LTLsquare \LTLdiamond g_2 \wedge  \LTLsquare( \neg w \rightarrow \LTLcircle (\neg g_1 \wedge \neg g_2)) \wedge \LTLsquare(\neg g_1 \vee \neg g_2) $$
This specification is realizable, with transition systems of size at least 4. We used our implementation to check whether we can find approximative solutions with smaller sizes. We used the tool CAQE~\cite{Rabe+Tentrup/15/CAQE}  for solving the QBF instances and the tool MaxCount~\cite{Fremont:EECS-2016-169} for solving the approximate synthesis instances.

The results are presented in Table~\ref{table:experiments}. As usual in synthesis, the size of the instances grows quickly as the size bound and number of processes increase. Inspecting the encoding constraints shows that the constraint for the specification is responsible for more than 80\% of the number of gates in the encoding. The results show that, using the approach we proposed, we can synthesize implementations for unrealizable specifications by bounding the environment. The results for the approximate synthesis method further demonstrate that for the unrealizable cases one can still obtain approximative implementations that satisfy the specification on a large number of input sequences.
	\begin{table}[t]
\centering
\scalebox{0.85}{
\begin{tabular}{|c||c|c|c|c||c|c|c|c||c|c|c|c|}
	\hline
	\multicolumn{5}{|c||}{instance} & \multicolumn{4}{c||}{QBF} & \multicolumn{4}{c|}{MaxCount} \\
	\hline
	Spec. & Proc. & \#States & Bound & Result & \#Gates & $\forall$ & $\exists$ & time & \#Max & \#Count & rate & time \\
	\hline
	\multirow{3}{1.8cm}{~~~Round-\\~~~Robin\\~~~Arbiter} & 2 & 2 & 4 & Unreal. & 15556 &  48& 12 & 9.91s & 12& 8& 0.5 & 26s\\
	\cline{2-13}
	  & 2 & 3 & 2 & Unreal. & 5338 & 40  & 24 & 2.45s & 24& 4& 0.88 & 161s\\ 
	\cline{2-13}
	 & 2 & 4 & 2 & Real. & 13414 & 60 & 12 & 12.15s & 40& 4& 0.88 & 283s\\
	\hline
	\hline
	 & 1 & 2 & 2 & Real. & 1597 & 20 & 10 & 0.41s & 10 & 4 & 1.0 & 0.79s \\
	\cline{2-13}
	& 1 & 2 & 3 & Unreal. & 4749 & 30 & 10 & 1.95s & 10 & 6 & 0.88 & 3.86s \\
	\cline{2-13}
	& 1 & 3 & 3 & Unreal. & 16861 & 48 & 21 & 17.26s & 21 & 6 & 0.88 & 20.83s \\
	\cline{2-13}
	Greedy& 1 & 4 & 3 & Real. & 43692	 &  78 & 36  & 3m7.44s & 36 & 6 & 1.0 & 2m43s \\
	\cline{2-13}
	Arbiter& 1 & 4 & 4 & - & 169829 & 104 & 36 & TO & 36 & 8& -& TO\\
	\cline{2-13}
	& 2 & 4 & 2 & Real. & 24688 & 62 & 72 & 1m.24s &72 & 6 & -& TO\\
	\cline{2-13}
	& 2 & 4 & 3 & Unreal. & 103433 & 93 & 72 & 27m15.2 &72 & 12& -& TO\\
	\cline{2-13}
	& 3 & 2 & 2 & Unreal. & 3985 & 93 & 72 & 1.39s & 38 & 8 & 0.65 & 4.18s\\
	\hline
\end{tabular}
}
\vspace{0.2cm}
\caption{\small Experimental results for the symbolic approaches. The rate in the approximate approach is the rate of input lassos on which the specification is satisfied.\looseness=-1}\label{table:experiments}
\end{table}

\section{Conclusion}
In many cases, the unrealizability of a specification is due to the assumption that the environment has unlimited power in producing inputs to the system.  In this paper, we have investigated the problem of synthesizing implementations under bounded environment behaviors. We have presented algorithms for solving the synthesis problem for bounded lassos and the synthesis of approximate implementations that satisfy the specification up to a certain rate. 

We have also provided polynomial encodings of the problems into quantified  Boolean formulas and maximum model counting instances. Our experiments demonstrate the principal feasibility of the approach. Our experiments also show that the instances can quickly become large. While this is a common phenomenon for synthesis, there clearly is a lot of room for optimization and experimentation with both the solvers for quantified Boolean expressions and for maximum model counting. 

\bibliographystyle{plain}
\bibliography{bib.bib}

\end{document}